\documentclass[letterpaper,10pt,conference]{ieeeconf}

\IEEEoverridecommandlockouts                              
\usepackage{graphicx} 
\usepackage{xcolor}
\usepackage{amsfonts}
\usepackage{amssymb}
\usepackage{graphicx}		
\usepackage{empheq}
\usepackage{tikz}
\usepackage{comment}
\usetikzlibrary{arrows.meta,positioning,calc}

\definecolor{plantblue}{RGB}{224,236,248}
\definecolor{controlorange}{RGB}{252,233,214}
\definecolor{sumfill}{RGB}{250,250,250}

\usepackage{pgfplots}       
\pgfplotsset{compat=1.18}
\usepackage[table]{xcolor}	
\let\labelindent\relax
\usepackage{enumitem}

\usepackage{amsmath}
\newcommand{\He}{\operatorname{He}}
\usepackage{amssymb}		
\usepackage[hidelinks]{hyperref}

\newtheorem{theorem}{Theorem}
\newtheorem{definition}{Definition}
\newtheorem{lemma}{Lemma}
\newtheorem{proposition}{Proposition}
\newtheorem{corollary}{Corollary}
\newtheorem{assumption}{Assumption}

\newcommand{\col}[1]{\operatorname{col}\!\left(#1\right)}

\usepackage{tikz}
\usetikzlibrary{arrows.meta,positioning,calc}

\definecolor{gridblue}{RGB}{63,96,125}
\definecolor{softblue}{RGB}{240,247,252}
\definecolor{boxblue}{RGB}{220,235,247}
\definecolor{accentteal}{RGB}{72,140,145}
\definecolor{textdark}{RGB}{35,35,35}
\definecolor{plantblue}{RGB}{222,235,247}
\definecolor{controlorange}{RGB}{252,232,206}
\definecolor{sumgray}{RGB}{245,245,245}
\definecolor{gridblue}{RGB}{190,110,35}    
\definecolor{softblue}{RGB}{255,243,230}   
\definecolor{boxblue}{RGB}{255,224,196}    
\definecolor{accentteal}{RGB}{220,125,35}  
\definecolor{textdark}{RGB}{70,45,25}

\title{Safety Verification for Voltage Control in Inverter Dominated Radial Distribution Networks under Actuator and Sensor Adversarial Attacks}
\author{Georgios Pantazis, Revi Boelen, Michelle S. Chong 
\thanks{*This work is funded through the CETPartnership’s ProRES project under the Joint Call 2024, co‑funded by the European Commission (Grant Agreement No. 101069750) and the Dutch Research Council NWO (File No. EP.1602.24.001).} 
\thanks{The authors are with the Dynamics and Control (D\&C) section at the Department of Mechanical Engineering, Eindhoven University of Technology, the Netherlands. 
Email: {\tt\small r.e.boelen@student.tue.nl,
\{G.Pantazis33, M.S.T.Chong\}@tue.nl}.}%
\thanks{This work has been submitted to the IEEE for possible
publication. Copyright may be transferred without notice,
after which this version may no longer be accessible.}
}

\begin{document}

\maketitle

\begin{abstract}
This paper develops a safety verification framework for
inverter-based radial distribution networks under sensor
and actuator attacks. Each customer is equipped with a droop-based controller to regulate its voltage within a desirable bound of a reference, i.e., to achieve safety. In our setting, attacks and
disturbances enter the droop-based controllers, while
actuator attacks corrupt the reactive power
injected into the network. Consequently, invariance of
the controller state alone does not guarantee compliance
with inverter capacities or customer voltage limits. To address this challenge,  we derive  conditions to guarantee that the reactive power state remains inside a predefined ellipsoid, while at the same time ensuring that the attacked reactive power injections
and customer voltages respect predetermined operating limits that hold uniformly for the entire predefined ellipsoid. Together, these conditions
guarantee safety under all admissible disturbances and  state-dependent sensor and actuator attacks.
Finally, we express the resulting containment condition as a linear matrix inequality (LMI) feasibility problem. A five-customer radial 
network based on the CIGRE benchmark corroborates our theoretical developments in simulations.
\end{abstract}

\section{Introduction}

High penetration of renewable energy generation by prosumers can
cause voltage deviations in power distribution systems, rendering the need for effective voltage regulation increasingly important. Local
inverter-based droop control, a form of Volt/VAR control, is a widely
studied approach to this problem in radial distribution networks \cite{baker2018,singhal2019}.
Each inverter adjusts its reactive power injection according to a
local droop function relating the measured voltage deviation
to a reactive power command with the aim of controlling the voltage. However, corrupted measurements can alter the droop controller response,
while actuator attacks can modify the reactive power
injected into the network. Safe operation therefore requires
guarantees on both customer voltages and inverter injections
under these attacks.
The DistFlow equations introduced in \cite{Baran_Wu} are used for studies of voltage regulation in radial distribution
networks, including local Volt/VAR control
\cite{farivar2013} and approaches based on communication
\cite{bolognani2019need,magnusson2020distributed}.
A broader overview is provided in
\cite{srivastava2023voltage}.
Voltage bounds under customer disturbances
are studied in \cite{chong2019}, which considers inverter-based radial
distribution networks with slope-restricted nonlinear droop controllers
and derives sufficient conditions for customer voltages to remain
within prescribed limits.  This is extended in~\cite{Chong2020} to account for a broader collection of sector-bounded droop controllers. More recently, \cite{pantazis2026heterogeneous} derives deterministic voltage certificates that retain the heterogeneity across connected prosumers and then exploit these certificates for the synthesis of a coordinated droop control architecture. However, the aforementioned works assume that the sensing and actuation channels used by the droop controllers are uncompromised.

\begin{figure}[t]
\centering
\begingroup

\def\prosumerrow#1#2#3#4#5{%
\begin{scope}[shift={(0,#2)}]

    \draw[dline, fill=softblue]
        (3.15,-0.58) rectangle (6.95,0.58);

    \node[
        anchor=south east,
        text=textdark
    ] at (6.95,0.68) {Customer $#1$};

    \node[box] (inv#1) at (4.40,0)
        {Inverter};

    \draw[line] (0,0) -- (inv#1.west);

    \draw[imp] (1.05,0) -- (1.68,0);

    \node[
        small,
        anchor=north,
        text=textdark
    ] at (1.365,-0.14) {$#3$};

    \draw[arr] (2.85,-0.27) -- (1.90,-0.27);

    \node[
        small,
        anchor=north,
        text=textdark
    ] at (2.375,-0.40) {$#5$};

    \node[
        small,
        anchor=south east,
        text=textdark
    ] at (3.05,0.05) {$#4$};

    \draw[line] (inv#1.east) -- (5.98,0);

    \draw[source] (6.20,0) circle (0.21);
    \node[tinylabel, text=textdark] at (6.20,0.11) {$+$};
    \node[tinylabel, text=textdark] at (6.20,-0.11) {$-$};

\end{scope}
}

\begin{tikzpicture}[
    scale=0.95,
    every node/.style={transform shape},
    font=\footnotesize,
    line/.style={
        draw=gridblue,
        line width=1.25pt
    },
    dline/.style={
        draw=gridblue,
        dashed,
        line width=1.0pt,
        rounded corners=1pt
    },
    arr/.style={
        draw=accentteal,
        ->,
        line width=1.1pt
    },
    imp/.style={
        draw=gridblue,
        line width=5.2pt,
        line cap=round
    },
    source/.style={
        draw=gridblue,
        line width=1.1pt,
        fill=white
    },
    box/.style={
        draw=gridblue,
        fill=boxblue,
        line width=1.0pt,
        rounded corners=2pt,
        minimum width=1.65cm,
        minimum height=0.78cm,
        align=center,
        text=textdark
    },
    small/.style={
        font=\scriptsize,
        inner sep=1pt
    },
    tinylabel/.style={
        font=\scriptsize,
        inner sep=0.5pt
    }
]

\draw[dline, fill=softblue]
    (-0.95,-0.50) rectangle (0.58,-6.85);

\draw[source, fill=white] (0,0) circle (0.34);

\draw[line]
    (-0.19,0)
    .. controls (-0.08,0.16) and (0.08,-0.16) ..
    (0.19,0);

\node[
    anchor=west,
    text=textdark
] at (0.56,0) {Substation};

\draw[line] (0,-0.34) -- (0,-3.35);

\draw[line] (0,-5.15) -- (0,-6.45);

\fill[gridblue] (0,-4.10) circle (1.25pt);
\fill[gridblue] (0,-4.38) circle (1.25pt);
\fill[gridblue] (0,-4.66) circle (1.25pt);

\node[
    small,
    anchor=east,
    text=textdark
] at (-0.10,-0.82) {$\hat v_0$};

\node[
    small,
    anchor=east,
    text=textdark
] at (-0.10,-1.70) {$\hat v_1$};

\node[
    small,
    anchor=east,
    text=textdark
] at (-0.10,-3.35) {$\hat v_2$};

\node[
    small,
    anchor=east,
    text=textdark
] at (-0.10,-5.35) {$\hat v_{N-1}$};

\node[
    small,
    anchor=east,
    text=textdark
] at (-0.10,-6.45) {$\hat v_N$};

\draw[imp] (0,-1.05) -- (0,-1.34);

\node[
    small,
    anchor=east,
    text=textdark
] at (-0.10,-1.195) {$Z_0$};

\draw[imp] (0,-2.25) -- (0,-2.56);

\node[
    small,
    anchor=east,
    text=textdark
] at (-0.10,-2.405) {$Z_1$};

\draw[imp] (0,-5.72) -- (0,-6.02);

\node[
    small,
    anchor=east,
    text=textdark
] at (-0.10,-5.87) {$Z_{N-1}$};

\draw[arr] (0.48,-0.98) -- (0.48,-1.48);

\node[
    small,
    anchor=west,
    text=textdark
] at (0.64,-1.23) {$P_0+jQ_0$};

\draw[arr] (0.48,-2.22) -- (0.48,-2.82);

\node[
    small,
    anchor=west,
    text=textdark
] at (0.64,-2.52) {$P_1+jQ_1$};

\draw[arr] (0.48,-5.68) -- (0.48,-6.22);

\node[
    small,
    anchor=west,
    text=textdark
] at (0.64,-5.95) {$P_{N-1}+jQ_{N-1}$};

\fill[gridblue] (0,-0.82) circle (1.6pt);

\fill[gridblue] (0,-1.70) circle (1.6pt);
\fill[gridblue] (0,-3.35) circle (1.6pt);

\fill[gridblue] (0,-5.35) circle (1.6pt);

\fill[gridblue] (0,-6.45) circle (1.6pt);

\prosumerrow{1}{-1.70}{Z'_0}{v_1}{\rho_1+jq_1}

\prosumerrow{2}{-3.35}{Z'_1}{v_2}{\rho_2+jq_2}

\node[text=textdark] at (4.85,-4.75) {$\vdots$};

\prosumerrow{N}{-6.45}{Z'_{N-1}}{v_N}{\rho_N+jq_N}

\draw[
    draw=gridblue,
    fill=softblue,
    line width=0.9pt,
    rounded corners=2pt
] (1.45,-9.05) rectangle (6.55,-7.20);

\node[
    anchor=west,
    text=textdark
] at (1.72,-7.42) {Legend:};

\draw[imp] (2.05,-7.85) -- (2.68,-7.85);

\node[
    small,
    anchor=west,
    text=textdark
] at (3.02,-7.85) {Line impedance};

\draw[source] (2.36,-8.30) circle (0.21);

\node[tinylabel, text=textdark]
    at (2.36,-8.19) {$+$};

\node[tinylabel, text=textdark]
    at (2.36,-8.41) {$-$};

\node[
    small,
    anchor=west,
    text=textdark
] at (3.02,-8.30) {DC voltage source};

\draw[source] (2.36,-8.70) circle (0.21);

\draw[
    draw=gridblue,
    line width=0.9pt
]
    (2.21,-8.70)
    .. controls (2.27,-8.57) and (2.45,-8.83) ..
    (2.51,-8.70);

\node[
    small,
    anchor=west,
    text=textdark
] at (3.02,-8.70) {AC voltage source};

\end{tikzpicture}
\endgroup

\caption{Radial network topology with $N$ prosumers.
$P_i$ and $Q_i$ denote the active and reactive power flows
along feeder.
The net active/reactive customer injections are
$\rho_i=\rho_{g,i}-\rho_{c,i}$ and
$q_i=q_{g,i}-q_{c,i}$, respectively. Voltage and impedance notation is defined in Section~II-A.}
\label{fig:lv-distribution-grid}
\end{figure}
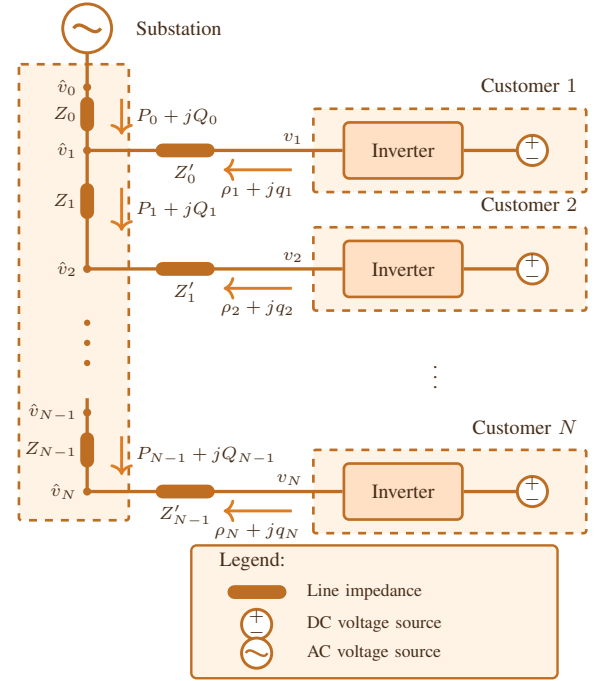

Cyberattacks affecting voltage regulation in distribution networks have been considered from many different perspectives \cite{wang2025sensor,pan2025,dimoulias2025}. The authors in \cite{Ma2017} consider falsification of the voltage measurements supplied to the droop controllers and characterize how such attacks affect stability and the customer voltage magnitudes. Their analysis demonstrates that corrupted local voltage measurements can propagate through the network and adversely affect buses other than the directly attacked one, but does not seek a safety certificate that holds at all times under a considered set of admissible attacks. The authors in \cite{Choeum2019} consider false data injection attacks against Volt/VAR control, where manipulated smart meter readings are designed to alter the voltages of the feeder towards abnormal operating conditions. Their aim is to construct effective stealthy attacks  rather than certify safety of the closed-loop system under all admissible attacks. The work ~\cite{liu2024cfdi} considers a setting where coordinated false data injection on both measurement data and control commands can occur and, taking the role of the adversary, determine attack strategies that increase network voltage deviations. 

Our contributions compared to the related literature are the following:

\begin{enumerate}
    \item The closest approaches to our work~\cite{chong2019,pantazis2026heterogeneous} provide guarantees under bounded customer disturbances originating from consumption and renewable energy generation. However, both works assume uncompromised sensing and actuation, while existing cyberattack studies~\cite{Ma2017} do not provide robust safety certificates tailored for nonlinear droop-controlled radial networks. We derive sufficient  conditions that jointly certify
customer voltage limits and inverter capacities under
simultaneous sensor and actuator attacks, while retaining
individual customer disturbance bounds and accounting
for deviations of the feeder voltage from the reference.
    \item Though our analysis builds on the invariant set construction and the so-called $\mathcal{S}$-procedure methodology leveraged in \cite{lin2026}, it addresses challenges that fall outside the formulation considered therein. Specifically, in our setting, sensor and actuator attacks, together with exogenous disturbances, enter directly into the nonlinear droop argument, while the physical quantities to be kept safe, i.e., the attacked inverter commands and customer voltages, depend explicitly on these signals, so their operating
limits cannot be guaranteed from state invariance alone. To address this, we derive a robust positive invariance condition for the reactive power state set, together with
new containment conditions that exploit the power distribution network model to certify inverter capacity and individual customer voltage limits for all admissible
attacks and disturbances. We then show that the obtained certificate reduces to a linear matrix inequality (LMI) feasibility problem for each fixed value of a single scalar parameter.
\end{enumerate}
\emph{Organization: }Section~II introduces the radial distribution network model,
the local droop control and attack models required to formulate the safety verification problem.
Section~III derives sufficient conditions for robust positive
invariance and containment within the operating limits, and reformulates these conditions as an LMI
feasibility problem.
Section~IV illustrates the proposed framework through numerical
simulations on a five-customer network.
Finally, Section~V concludes the paper and proposes further research directions.

\section{Problem Formulation}
\label{sec:problem}

\emph{Notation:} Let \(\mathbb R^N_{\ge 0}\) and \(\mathbb R^N_{>0}\) denote the
nonnegative and positive orthants of \(\mathbb R^N\), respectively.
 The vector of all
ones is denoted by \(\mathbf 1\), and \(\operatorname{diag}(x)\)
denotes the diagonal matrix with diagonal entries given by the vector
\(x\).  \(A_{ij}\) denotes the \((i,j)\)-th entry of matrix $A$.  For a symmetric matrix \(A\), the
relations \(A\succ0\), \(A\succeq0\), \(A\prec0\), and \(A\preceq0\)
denote positive definite, positive semidefinite, negative definite, and
negative semidefinite, respectively. For matrices of compatible dimensions, denote $\text{He}(X):=X+X^\top.$
The operator $\col{x_1,\ldots,x_m}$ denotes vertical
concatenation of its arguments, while the $i$-th canonical basis
vector of $\mathbb R^N$ is denoted by $e_i$. Finally, let $e\in\mathbb R^{5N+1}$
denote the last canonical basis vector, i.e.,
    $e:=\col{0_N,0_N,0_{2N},0_N,1}$. For a set \(\mathcal S\subseteq\mathbb R^N\), its boundary is denoted by
\(\partial\mathcal S\). Expressions containing $\pm$ are imposed for
both signs, with separate decision variables.

\subsection{ Radial Network Model and Nominal Droop Control} \label{sec:problem.A}

We consider a model of a radial low voltage distribution grid with $N$ customers, as shown in
Figure~1, based on the DistFlow equations \cite{Baran_Wu}. Each customer $i \in \mathcal{N}=\{1, \dots, N\}$ has an inverter which can inject active and reactive power $\rho_{g, i}$ and $q_{g,i}$, respectively. 
Let $v_i\geq0$ denote the voltage magnitude at customer $i$,
and $\hat v_i\geq0$ the voltage magnitude at the corresponding
connection point on the feeder.  The feeder segment between connection points \(i\) and \(i+1\) has impedance $Z_i = R_i + j X_i, i \in  \mathcal{N}_{-}=\{0, \dots, N-1\}$, while the impedance of the line between customer $i$ and its connection point is $Z'_{i-1} = R'_{i-1} + j X'_{i-1}, i \in \mathcal{N}$. Let $\rho_{c,i}$ and $q_{c,i}$ denote the active and
reactive power consumption of customer $i$, respectively,
and define the net active power injection
$\rho_i:=\rho_{g,i}-\rho_{c,i}$.
For a reference voltage $\bar v>0$, we consider the squared voltage differences $y_i:=\bar v^2-v_i^2$
    and $y_0:=\bar v^2-\hat v_0^2$,
where $\hat v_0$ is the voltage at the head of the feeder. The $i$-th droop controller uses $y_i$ to adjust the reactive
power command $q_{g,i}$ and regulate the $i$-th customer's voltage $v_i$. For compactness, we define the vectors $q_g:=\operatorname{col}(q_{g,1},\ldots,q_{g,N})$,
$\rho:=\operatorname{col}(\rho_1,\ldots,\rho_N)$,
$q_c:=\operatorname{col}(q_{c,1},\ldots,q_{c,N})$, and
$y:=\operatorname{col}(y_1,\ldots,y_N) \in \mathbb R^N$. Following~\cite{chong2019,    pantazis2026heterogeneous}, the linearized
DistFlow equations  yield for each customer $i \in \mathcal{N}$:
\begin{align}
    y_i=(Hq_g)_i+[\phi(\rho,q_c)]_i+y_0,
\end{align}
where $(Hq_g)_i \in \mathbb{R}$ denotes the $i$-th component of the vector $Hq_g \in \mathbb{R}^N$
and 
 $H=(h_{ij})_{i,j\in\mathcal N}$ is defined as:
\begin{align}
H
=
-2
\begin{bmatrix}
X_0 & X_0 & \cdots & X_0 \\
X_0 & X_0+X_1 & \cdots & X_0+X_1 \\
\vdots & \vdots & \ddots & \vdots \\
X_0 & X_0+X_1 & \cdots & \displaystyle\sum_{k=0}^{N-1}X_k \nonumber 
\end{bmatrix} \\
-
2\operatorname{diag}(X'_0,\dots,X'_{N-1}), \nonumber 
\end{align}
while \(\phi(\rho,q_c)\) collects all contributions of net active power injections and consumptions. As derived in \cite{chong2019}, the $i$-th component of its analytical form is
\begin{align}
    [\phi(\rho,q_c)]_i
    =
    \sum_{k=0}^{i-1}
    \left(
        2X_k \sum_{j=k+1}^{N} q_{c,j}
        -
        2R_k \sum_{j=k+1}^{N} \rho_j
    \right) \nonumber \\
    -2R'_{i-1}\rho_i
+2X'_{i-1}q_{c,i}. \nonumber 
\end{align}
Thus, the state and output dynamics can be written in compact form as follows:
\begin{subequations}
\begin{empheq}[left=\empheqlbrace]{align}
\dot q_g &= -A q_g + A K(y),
\label{state_dynamics}\\
y &= Hq_g + \phi(\rho,q_c) + y_0\mathbf{1}.
\label{output_dynamics}
\end{empheq}
\end{subequations}
where  $K(y) := (K_1(y_1),\dots,K_N(y_N))^\top$,  $A=\text{diag}(\tau_1^{-1},\ldots,\tau_N^{-1}) \succ 0$, with $\tau_i > 0$ being the time constant of each inverter $i$. 
We impose the following assumption. \\
\begin{assumption} \label{assum:sector-bounded}  The nonlinear droop function $K(y)=(K_1(y_1), \dots, K_N(y_N))^\top$ satisfies 
\begin{align}
\frac{K_i(w)-K_i(v)}{w-v} \in [0, d_i], \nonumber 
\end{align}
for all $w,v \in \mathbb{R}$ with $w \neq v$ 
and $K_i(0)=0$ for each $i \in \mathcal{N}$. Furthermore, the bound $|K_i(v)| \leq \bar{K}_i$ holds  for all $v \in \mathbb{R}$ and each $i \in \mathcal{N}$. \hfill $\square$
\end{assumption}
Assumption \ref{assum:sector-bounded} is standard in droop  control implementations.  A commonly used example satisfying Assumption~\ref{assum:sector-bounded} is the piecewise-affine droop function used in
\cite{chong2019,zhou2021reverse}. Similar to \cite{chong2019}, we choose the design parameters such that $y_{\min,i}<y_{m,i}\leq 0\leq y_{n,i}<y_{\max,i}$.
The droop function is then defined as
\begin{align}
K_i(y)=
\begin{cases}
-\bar K_i,
    & y\leq y_{\min,i},\\[1mm]
-\left(
1-\dfrac{y-y_{\min,i}}
{y_{m,i}-y_{\min,i}}
\right)\bar K_i,
    & y_{\min,i}<y\leq y_{m,i},\\[3mm]
0,
    & y_{m,i}<y\leq y_{n,i},\\[1mm]
\dfrac{y-y_{n,i}}
{y_{\max,i}-y_{n,i}}\bar K_i,
    & y_{n,i}<y\leq y_{\max,i},\\[3mm]
\bar K_i,
    & y>y_{\max,i}.
\end{cases} \label{eq:droop_piecewise}
\end{align}
The interval $(y_{m,i},y_{n,i}]$ is the so-called dead zone, within which no
reactive power response is requested. Outside the dead zone, the droop
output varies affinely with the voltage error until it reaches the
saturation values $-\bar K_i$ and $\bar K_i$\footnote{When the controller is
implemented at customer $i$, the scalar argument $y$ is evaluated at
$y=y_i=\bar v^2-v_i^2$.}. We now impose the following assumptions: 
\begin{assumption}\label{ass:prosumer_bounds}
For each customer $i\in\mathcal N$, there exist constants $\bar\rho_i\geq 0$ and $\bar q_{c,i}\geq 0$ such that $|\rho_i(t)|\leq\bar\rho_i$, and 
    $|q_{c,i}(t)|\leq\bar q_{c,i}$, for all $t\geq 0$. \hfill $\square$
\end{assumption} 
\begin{assumption}\label{ass:head_voltage}
There exists a constant $\Delta_0\geq 0$ such that $|\hat v_0(t)-\bar v|\leq\Delta_0$ for all $t\geq 0$. \hfill $\square$
\end{assumption}
Under Assumption~\ref{ass:prosumer_bounds}, the following
result provides heterogeneous bounds on the exogenous
term $\phi(\rho,q_c)$.
\begin{lemma}[Adapted from {\cite[Lemma~1]{pantazis2026heterogeneous}}]
\label{lem:heterogeneous_disturbance}
Consider Assumption \ref{ass:prosumer_bounds}. Then,  for any $i \in \mathcal{N}$, it holds that $\bigl|[\phi(\rho(t),q_c(t))]_i\bigr|
    \leq \bar\phi_i$  for all $t \geq 0$
where
\begin{equation}\label{eq:phi_bound}
\begin{aligned}
    \bar\phi_i
    := {}&
    2\sum_{k=0}^{i-1}
    \left(
        R_k\sum_{j=k+1}^{N}\bar\rho_j
        +
        X_k\sum_{j=k+1}^{N}\bar q_{c,j}
    \right)\\
    &+2R_{i-1}'\bar\rho_i
     +2X_{i-1}'\bar q_{c,i}. 
\end{aligned}
\end{equation}
\end{lemma}

\emph{Proof:}
Applying the triangle inequality to $[\phi(\rho,q_c)]_i$
and substituting the bounds in
Assumption~\ref{ass:prosumer_bounds}
gives~\eqref{eq:phi_bound}. \hfill $\blacksquare$

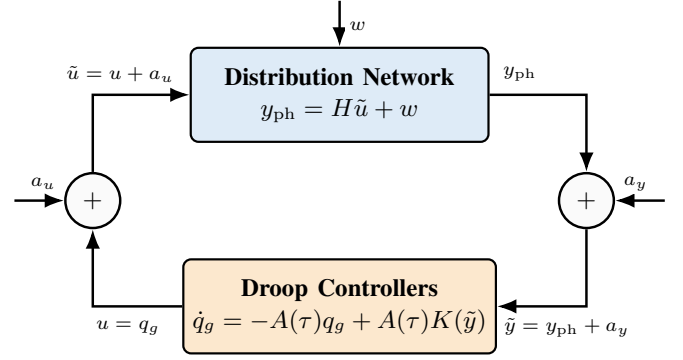
\begin{figure}[t]
\centering
\resizebox{\columnwidth}{!}{%
\begin{tikzpicture}[
    >=Latex,
    line width=0.9pt,
    font=\small,
    block/.style={
        draw,
        thick,
        rounded corners=1mm,
        minimum width=3.8cm,
        minimum height=1.2cm,
        align=center
    },
    plant/.style={block, fill=plantblue},
    controller/.style={block, fill=controlorange},
    sum/.style={
        draw,
        circle,
        fill=sumfill,
        minimum size=7mm,
        inner sep=0pt
    },
    wirelabel/.style={
        fill=white,
        inner sep=1.5pt,
        font=\scriptsize
    }
]

\node[plant]      (plant) at (0, 1.35)
{\textbf{Distribution Network}\\[0.5mm] $y_{\mathrm{ph}} = H\tilde{u} + w$};

\node[controller] (ctrl)  at (0,-1.35)
{\textbf{Droop Controllers}\\[0.5mm] $\dot{q}_g = -A(\tau)q_g + A(\tau)K(\tilde{y})$};

\node[sum] (sumu) at (-3.15,0) {$+$};
\node[sum] (sumy) at ( 3.15,0) {$+$};

\draw[->] (ctrl.west) -- (-3.15,-1.35) -- (sumu.south);

\draw[->] (sumu.north) -- (-3.15,1.35) -- (plant.west);

\draw[->] (plant.east) -- (3.15,1.35) -- (sumy.north);

\draw[->] (sumy.south) -- (3.15,-1.35) -- (ctrl.east);

\draw[->] (0,2.55) -- (plant.north);
\draw[->] (-4.15,0) -- (sumu.west);
\draw[->] ( 4.15,0) -- (sumy.east);

\node[wirelabel] at (-2.7,-1.62) {$u=q_g$};
\node[wirelabel] at (-2.8,1.62) {$\tilde{u}=u+a_u$};

\node[wirelabel] at ( 2.3, 1.62) {$y_{\mathrm{ph}}$};
\node[wirelabel] at ( 2.9,-1.62) {$\tilde{y}=y_{\mathrm{ph}}+a_y$};

\node[wirelabel] at (0.23,2.23) {$w$};
\node[wirelabel] at (-3.78,0.22) {$a_u$};
\node[wirelabel] at ( 3.78,0.22) {$a_y$};

\end{tikzpicture}%
}
\caption{Closed-loop distribution network with droop control
under actuator and sensor attacks.}
\label{fig:attacked_system}
\end{figure}

\begin{corollary}\label{cor:disturbance_bounds}
Under Assumptions~\ref{ass:prosumer_bounds}
and~\ref{ass:head_voltage}, the disturbance
$w=\phi(\rho,q_c)+(\bar v^2-\hat v_0^2)\mathbf{1}$
satisfies $|w_i(t)|\leq\bar w_i$ for all $i\in\mathcal N$, and $t\geq 0$, where
    $\bar{w}_i=\bar\phi_i+2\bar v\Delta_0+\Delta_0^2$, with  $\bar\phi_i$ as in 
Lemma~\ref{lem:heterogeneous_disturbance}. \hfill $\square$
\end{corollary}
\emph{Proof:}
By Assumption~\ref{ass:head_voltage}, $|\bar v^2-\hat v_0(t)^2| \leq
    2\bar v|\hat v_0(t)-\bar v|
    +|\hat v_0(t)-\bar v|^2
    \leq 2\bar v\Delta_0+\Delta_0^2$.
Leveraging 
Lemma~\ref{lem:heterogeneous_disturbance} and the
triangle inequality concludes the proof. \hfill $\blacksquare$ \\
In other words, it holds that $w \in \mathcal{W}$, where we define $\mathcal{W}$ as
\begin{align}
\mathcal W
    :=
    \left\{
        w\in\mathbb R^N:
        |w_i|\leq\bar w_i, \text{ for each } i\in\mathcal N
    \right\}.  \nonumber 
\end{align} 
    The control objective is to maintain the voltage magnitude at every customer
within a prescribed operating range around the nominal voltage $\bar v$.
Specifically, for a given tolerance $\delta \geq 0$, the desired safety
requirement is
\begin{equation}
\label{eq:systemsafety}
    \bar v-\delta \leq v_i(t) \leq \bar v+\delta,  
\end{equation}
for all $i\in \mathcal{N}$, for all  $t\geq0$.
Under nominal operation, the droop controllers decide on the inverter
reactive power injections using the locally measured voltages. However,
corruption of either the voltage measurements or the actuator commands can have  a significant effect on the closed-loop response, thus driving the physical voltages
outside the safe range in \eqref{eq:systemsafety}. We therefore seek
conditions under which \eqref{eq:systemsafety} as well as the
inverter capacity limits are satisfied for all
admissible sensor and actuator attacks and exogenous disturbances considered. This
motivates the closed-loop model under input and output attacks introduced in the next subsection. 

\subsection{ Disturbances and Attack Models}
We now extend the nominal model in \eqref{state_dynamics}-\eqref{output_dynamics}
to account for malicious corruption of the sensing and actuation channels.
For compactness, we denote the exogenous disturbance by $w := \phi(\rho,q_c) + y_0\mathbf{1}$,
such that the output equation is 
\begin{equation}
    y = Hq_g+w.
\end{equation}

Let $a_u\in\mathbb R^N$ denote an additive actuator attack.
The nominal command is $u=q_g$, whereas the network
receives the attacked injection $\tilde u=q_g+a_u$.
Consequently, inverter capacity limits must be imposed
on $\tilde u$, rather than on $q_g$ alone. Under an additive attack $a_y\in\mathbb{R}^N$ on the squared voltage
deviation measurement, the output dynamics in \eqref{output_dynamics} take the form  $\tilde y= Hq_g+Ha_u+w+a_y$. Consequently, replacing the nominal controller input $y$ in
\eqref{state_dynamics} by $\tilde y$ yields
\begin{align}
    \dot q_g
    &= -Aq_g + AK(\tilde y) \notag\\
    &= -Aq_g
       +AK\!\left(Hq_g+w+a_y+Ha_u\right).
\end{align} 

Denoting $a:=\operatorname{col}(a_y,a_u)$,
$D:=[I\ H]$, and $D_y:=[0\ H]$, we have
$Da=a_y+Ha_u$ and $D_ya=Ha_u$.
Here, $D_y$ captures the effect of actuator attacks on
the actual customer voltages. Thus,
$y_{\mathrm{ph}}:=Hq_g+w+D_ya$ denotes the squared
voltage deviations at the customers, including the effect
of actuator attacks but before sensor corruption.
The corrupted voltage measurement $\tilde{y}$ available to the
droop controller satisfies $\tilde y=y_{\mathrm{ph}}+a_y=Hq_g+w+Da.$
Thus,  the closed-loop system under input and output attacks can be written compactly in the form:

\begin{subequations}
\label{eq:attacked-system}
\begin{empheq}[left=\empheqlbrace]{align}
\dot q_g &= -Aq_g + AK(\tilde{y}),
\label{eq:attacked-state}\\
\tilde{y} &= Hq_g + w + D a.
\label{eq:attacked-output}
\end{empheq}
\end{subequations}
The disturbance $w=\phi(\rho,q_c)+y_0\mathbf 1$
collects the customer power disturbances and voltage
deviations from the reference at the head of the feeder. Following \cite{lin2026}, we describe admissible
attacks through a state-dependent quadratic constraint.

\begin{assumption}\label{ass:attack_set}
At each time $t\geq0$, the attack satisfies
$a(t)\in\mathcal A(q_g(t))$, where
\begin{equation}
\mathcal A(q_g):=
\left\{
a\in\mathbb R^{2N}:
\begin{bmatrix}q_g\\a\end{bmatrix}^{\!\top}
\begin{bmatrix}Q_1&Q_2\\Q_2^\top&Q_3\end{bmatrix}
\begin{bmatrix}q_g\\a\end{bmatrix}
\leq1
\right\},
\label{eq:attackset}
\end{equation}
with $Q_1=Q_1^\top\in\mathbb R^{N\times N}$,
$Q_2\in\mathbb R^{N\times2N}$, and
$Q_3=Q_3^\top\in\mathbb R^{2N\times2N}$ with
$Q_3\succ0$.
\hfill $\square$
\end{assumption}
Since $Q_3\succ0$, completing the square in  \eqref{eq:attackset} means that $\mathcal A(q_g)$ is nonempty if and only if
$q_g^\top
\left(Q_1-Q_2Q_3^{-1}Q_2^\top\right)q_g\le1.$
For a candidate ellipsoid
$\mathcal E(P):=\{q_g:q_g^\top Pq_g\le1\}$,
with $P=P^\top\succ0$, nonemptiness of
$\mathcal A(q_g)$ for every $q_g\in\mathcal E(P)$
is therefore equivalent to the inequality:
\begin{equation}
Q_1-Q_2Q_3^{-1}Q_2^\top\preceq P. 
\label{eq:attack-nonempty}
\end{equation}
We impose this condition on the candidate ellipsoids
used in the subsequent safety certificates.
\section{Safety Verification Under Attacks}
We derive the safety certificate in three steps. We first represent the characteristics of the droop controller  and the
admissible attacks and disturbances by quadratic constraints. We then establish conditions for robust positive
invariance of a set of reactive power states of the
closed-loop system in \eqref{eq:attacked-system}, which is
characterized by a candidate ellipsoid. In other words, when a reactive power state starts from within this ellipsoid, it is guaranteed that it will never leave it under the considered assumptions.  Finally, we impose conditions ensuring that the inverter commands and voltages remain within their predefined operating limits for all admissible attacks and disturbances.

Let $d_i$ denote the slope bound in
Assumption~\ref{assum:sector-bounded}.
The following lemma expresses the resulting sector constraint.
\begin{lemma} \label{lem:sector1}
Under Assumption~1, let $S_1=0$,
$S_2=\text{diag}(d_1,\ldots,d_N)$.
Then, for every positive diagonal matrix
$V=\text{diag}(\nu_1,\ldots,\nu_N)\succ0$, it holds for all $\eta \in \mathbb{R}^N$ that  
\begin{equation}
\label{eq:sector}
(K(\eta)-S_1\eta)^\top
    V (K(\eta)-S_2\eta)
    \leq 0.
    \end{equation} \label{eq:sector_valued}
 \hfill $\square$
\end{lemma}
\emph{Proof:}
For $\eta_i\neq0$, setting $v=0$ and $w=\eta_i$ in Assumption~1 gives $\dfrac{K_i(\eta_i)}{\eta_i} \in [0, d_i]$,
where $K_i(0)=0$. Hence, it holds that $K_i(\eta_i)
\bigl(K_i(\eta_i)-d_i\eta_i\bigr)\leq 0$.
For $\eta_i=0$, the inequality trivially holds since also
$K_i(0)=0$. Multiplying by $\nu_i>0$ and summing over
$i\in \mathcal{N}$ concludes then the proof. \hfill $\blacksquare$

Lemma 2 provides a quadratic inequality satisfied
by every droop controller that follows Assumption 1.
We use this inequality as an outer approximation
of the nonlinear relation. 
For the attacked system, define
$\gamma:=Hq_g+w+Da$ and $\kappa:=K(\gamma)$.  

We now define the augmented vector $\chi:=\col{q_g,\kappa,a,w,1}$. Furthermore, let us introduce the selection matrix $T:=
\begin{bmatrix}
H & 0 & D & I & 0\\
0 & I & 0 & 0 & 0
\end{bmatrix}$ such that
$\begin{bmatrix}
\gamma\\
\kappa
\end{bmatrix}=T\chi$.
The following result expresses \eqref{eq:sector}
as a quadratic constraint in $\chi$.

\begin{lemma} \label{lem:sector_quadratic}
Under  Assumption~\ref{assum:sector-bounded}, \eqref{eq:sector} evaluated at $\eta=\gamma$ is equivalent to the constraint:
\begin{align}
F_s(\chi):=\chi^\top\Phi_s\chi\leq0, \label{eq:sector-lmi}
\end{align}
with $\Phi_s \!=\!T^\top\Psi_sT$,  $\Psi_s \!=
\begin{bmatrix}
\text{He}(S_1^\top VS_2)
&
-(S_1+S_2)^\top V
\\
\star
&
2V
\end{bmatrix}.$
\end{lemma}
\emph{Proof:} Expanding twice the left-hand side of
$(\kappa-S_1\gamma)^\top
    V
    (\kappa-S_2\gamma)
    \leq 0$ gives
\begin{align}
&2(\kappa-S_1\gamma)^\top
V(\kappa-S_2\gamma) \nonumber \\
&=
2\kappa^\top V\kappa
-\kappa^\top V(S_1+S_2)\gamma \nonumber \\
&-\gamma^\top(S_1+S_2)^\top V\kappa
+\gamma^\top\text{He}(S_1^\top VS_2)\gamma
=
\begin{bmatrix}
\gamma\\ \kappa
\end{bmatrix}^{\!\top}
\Psi_s
\begin{bmatrix}
\gamma\\ \kappa
\end{bmatrix}. \nonumber 
\end{align}
Using $\begin{bmatrix}
        \gamma \\ \kappa
    \end{bmatrix} \nonumber 
    =T\chi$, it then holds that
\begin{equation}
2(\kappa-S_1\gamma)^\top
V(\kappa-S_2\gamma)
=
\chi^\top T^\top\Psi_sT\chi. \nonumber
\end{equation}
 which concludes the proof since 
$\Phi_s=T^\top\Psi_sT$.
\hfill$\blacksquare$

Lemma~\ref{lem:sector_quadratic} expresses the sector
constraint implied by Assumption~\ref{assum:sector-bounded}
in quadratic form of the augmented vector $\chi$.
We next combine this constraint with quadratic inequalities of the attack and disturbance bounds, aiming to
derive a sufficient condition for  robust positive
invariance.

\subsection{Robust Positive Invariance and Capacity Limits}

We first establish robust positive invariance of a
candidate set of the reactive power state.
We subsequently impose containment conditions ensuring
that the corresponding inverter reactive power injections and customer
voltages satisfy their operating limits.
The following definition formalizes the invariance property.
\begin{definition}
Consider the attacked closed-loop system in \eqref{eq:attacked-state}-\eqref{eq:attacked-output},
where $w(t)\in\mathcal W$ and \(a(t)\in\mathcal A(q_g(t))\) for all \(t\geq0\). A set \(\mathcal S\subseteq\mathbb R^N\) is said to be robust positively invariant (RPI) if, for every initial condition \(q_g(0)\in\mathcal S\) and every admissible disturbance and attack pair $(w(\cdot),a(\cdot)) \in \mathcal{W} \times \mathcal{A}(q_g(\cdot))$, the corresponding closed-loop trajectory satisfies $q_g(t)\in\mathcal S$ for all $t\geq0$.  \hfill $\square$
\end{definition} 

To construct an RPI set for our problem setting, we consider the quadratic
function  $V_q(q_g):=q_g^\top Pq_g$, with a symmetric and positive definite matrix $P$
and the ellipsoid of the form
\begin{equation}
    \mathcal E(P):=
    \left\{
        q_g\in\mathbb R^N:
        q_g^\top Pq_g\leq1
    \right\}.
    \label{eq:ellipsoid}
\end{equation}
Let us now define
the selection matrices
$E_q:=
\begin{bmatrix}
I&0&0&0&0
\end{bmatrix}$,
$E_w:=
\begin{bmatrix}
0&0&0&I&0
\end{bmatrix}$,
and $G:=
\begin{bmatrix}
-A&A&0&0&0
\end{bmatrix}.$
Leveraging those matrices, we can then write $q_g$ and $\dot q_g$ as an affine function of $\chi$. Specifically, we have that $q_g=E_q\chi$,  and  $\dot q_g=G\chi$.
Taking the time derivative of $V_q$  it holds that
\begin{align}
\dot V_q
&=
\dot q_g^\top Pq_g+q_g^\top P\dot q_g \notag\\
&=
\chi^\top
\text{He}(E_q^\top PG)\chi
    =\chi^\top\Phi_0(P)\chi:=F_0(\chi), \nonumber 
\end{align}
 where 
    $\Phi_0(P):=\text{He}(E_q^\top PG)$. The boundary of the ellipsoid $\mathcal E(P)$ can then be described by the equality:
\begin{align}
    F_b(\chi)&:=
    q_g^\top Pq_g-1
    = \nonumber \\
    &\chi^\top (E_q^\top PE_q-ee^\top) \chi=
    \chi^\top\Phi_{b}(P)\chi=0, \label{eq:ellip_boundary}
\end{align}
where $\Phi_{b}(P):=
    E_q^\top PE_q-ee^\top$.

Now let us define the selection matrix $T_a:= \begin{bmatrix}
I&0&0&0&0\\
0&0&I&0&0
\end{bmatrix}$, for which it holds that $T_a\chi=
    \begin{bmatrix}
        q_g &  a
    \end{bmatrix}^\top$. 
We now wish to render the disturbance and attack bounds  as quadratic constraints. For notational convenience, let us now introduce the matrix
   $Q_a:=
    \begin{bmatrix}
        Q_1&Q_2\\
        Q_2^\top&Q_3
    \end{bmatrix}$. 
\begin{lemma}
The following results hold: 
\begin{enumerate}
    \item[(i)] The constraint in \eqref{eq:attackset} can equivalently be written as
    \begin{align}
         F_a(\chi):=
        \chi^\top\Phi_a\chi\leq0, \label{eq:attack_constraint_lmi}
    \end{align}
    where 
    $\Phi_a:=T_a^\top Q_aT_a-ee^\top.$
\item[(ii)]  For each \(j\in\mathcal N\), let us define $\Phi_{w,j}
    :=
    E_w^\top e_j e_j^\top E_w
    -\bar w_j^2 ee^\top$,
    and $F_{w,j}(\chi)
    := \chi^\top \Phi_{w,j}\chi.$
    Then, $w\in\mathcal W$ holds if and only if 
    \begin{align}
    F_{w,j}(\chi)\leq 0, \text{ for all } j\in\mathcal N. \label{eq:disturbance_LMI}
    \end{align}
\end{enumerate}
\end{lemma}

\begin{proof}
(i) Since $T_a\chi=
\begin{bmatrix}
q_g & 
a
\end{bmatrix}^\top$
and \(e^\top\chi=1\), it follows that
\[
F_a(\chi)
=
\begin{bmatrix}
q_g\\
a
\end{bmatrix}^{\!\top}
Q_a
\begin{bmatrix}
q_g\\
a
\end{bmatrix}
-1.
\]
Hence, \(F_a(\chi)\leq0\) is equivalent to
\eqref{eq:attackset}.

(ii) Since \(E_w\chi=w\) and \(e^\top\chi=1\), it follows that $F_{w,j}(\chi)
=
w_j^2-\bar w_j^2$.
Hence, $F_{w,j}(\chi)\leq0$ for all $j \in \mathcal{N}$ holds 
if and only if $|w_j|\leq\bar w_j$ for all $j \in \mathcal{N}$ which is equivalent to $w\in\mathcal W$.
\end{proof}
Thus, every admissible closed-loop trajectory satisfies
the quadratic constraints in \eqref{eq:sector-lmi},
\eqref{eq:attack_constraint_lmi}, and
\eqref{eq:disturbance_LMI}.
On the boundary of $\mathcal E(P)$, we additionally
have $F_b(\chi)=0$, as expressed in
\eqref{eq:ellip_boundary}.
The following proposition now combines these relations
to obtain a sufficient condition for
$\dot V_q\leq0$ on the boundary, thereby ensuring
that trajectories initialized in $\mathcal E(P)$
remain in this set.
\begin{proposition} \label{prop:RPI}
Under Assumptions~1-4, consider the ellipsoid $\mathcal E(P)$ in
\eqref{eq:ellipsoid}. Furthermore, assume that \eqref{eq:attack-nonempty} holds. Then, 
if there exists a positive diagonal matrix
$V=\operatorname{diag}(\nu_1,\ldots,\nu_N)\succ0$,
and parameters $(\alpha_1, \alpha_2, (\alpha_{w,j})_{j \in \mathcal{N}}) \in \mathbb R \times \mathbb R_{+} \times \prod_{j \in \mathcal{N}}\mathbb R_{+}$
such that, for every $\chi\in\mathbb R^{5N+1}$,
\begin{equation}
\begin{aligned}
    F_0(\chi)+\alpha_1F_b(\chi)
    -\alpha_2&F_a(\chi)-F_s(\chi)\\
    &-\sum_{j=1}^{N}\alpha_{w,j}F_{w,j}(\chi)\leq0,
\end{aligned}
\label{eq:rpi-condition}
\end{equation}
then $\mathcal E(P)$ is RPI for the system in
\eqref{eq:attacked-state}-\eqref{eq:attacked-output}.
\hfill $\square$
\end{proposition}
\emph{Proof:} Consider any reactive power vector $q_g\in\partial\mathcal E(P)$. As such, it holds that  $F_b(\chi)
=
\chi^\top\Phi_b(P)\chi
=
q_g^\top Pq_g-1
=
0.$
Moreover, the constraints in \eqref{eq:sector-lmi},
\eqref{eq:attack_constraint_lmi}, and
\eqref{eq:disturbance_LMI}
imply, respectively, that $F_s(\chi)=\chi^\top \Phi_s \chi$, $F_a(\chi)$,
and $F_{w,j}(\chi)$ for all $j\in\mathcal N$ are nonpositive for every augmented vector $\chi$, associated with the closed-loop system
\eqref{eq:attacked-state}-\eqref{eq:attacked-output},
under admissible attacks and disturbances.
Since $\alpha_2,\alpha_{w,j}\geq0$, it holds that
\[
-\alpha_2F_a(\chi)
-F_s(\chi)
-\sum_{j=1}^{N}\alpha_{w,j}F_{w,j}(\chi)
\geq0.
\]
Together with $F_b(\chi)=0$ on $\partial\mathcal E(P)$,
\eqref{eq:rpi-condition} therefore yields
$\dot V_q(q_g)=F_0(\chi)\leq0$
for every admissible attack and disturbance whenever
$q_g\in\partial\mathcal E(P)$. Hence, the vector field governing $q_g$ in
\eqref{eq:attacked-system} cannot point
outward through the boundary of $\mathcal E(P)$, and therefore
$\mathcal E(P)$ is robust positively invariant. \hfill $\blacksquare$

Proposition 1 guarantees that the controller state $q_g$ remains in $\mathcal E(P)$ for every attack satisfying Assumption~\ref{ass:attack_set}
and every disturbance $w\in\mathcal W$. This property alone does not, however,  establish the satisfaction of the inverter capacity or the predefined voltage limits. In the subsequent developments, we  derive conditions ensuring that these operating limits are satisfied  throughout the invariant ellipsoid $\mathcal E(P)$.
Condition~\eqref{eq:rpi-condition} depends on the network and droop controller parameters through $\Phi_s$, while the bounds $\bar\rho_i$, $\bar q_{c,i}$, and $\Delta_0$ affect $\Phi_{w,j}$. Likewise, the admissible sensor and actuator attacks
$a=\operatorname{col}(a_y,a_u)$ are characterized by the matrices
$Q_1$, $Q_2$, and $Q_3$ and determine the attack constraint
$F_a(\chi)=\chi^\top\Phi_a\chi\leq0$. For a given operating scenario, feasibility of \eqref{eq:rpi-condition} certifies an RPI set $\mathcal E(P)$ for all disturbances $w\in\mathcal W$ and attacks $a\in\mathcal A(q_g)$. As such, this condition is sufficient to verify whether a given uncertainty and/or attack can be handled by the closed-loop control system. Conversely, infeasibility of \eqref{eq:rpi-condition} only indicates that an RPI set of the considered form could not be certified and does not, by itself, imply that the closed-loop system is unsafe. 

Though Proposition \ref{prop:RPI} ensures that   the reactive power $q_g$ remains within the ellipsoid $\mathcal E(P)$,
the attacked actuator command $\tilde u=q_g+a_u$ is still not guaranteed to
respect the inverter capacity limits. Ensuring the satisfaction of these limits prevents hardware saturation from
altering the injection assumed in the network model.
To establish this guarantee, let $E_u:=[0\ I]$ and
$C_u:=[I\ 0\ E_u\ 0\ 0]$, so that $\tilde u=C_u\chi$. Let us then define 
\[
\Phi_{u,i}:=
\frac{1}{\bar Q_i^2}C_u^\top e_i e_i^\top C_u-ee^\top, 
\]
which allows us to bring the capacity constraint $|\tilde u_i|\leq\bar Q_i$ in the form
$F_{u,i}(\chi):=\chi^\top\Phi_{u,i}\chi\leq0$.
The following result ensures this constraint is satisfied for any reactive power $q_g \in \mathcal E(P)$ under every admissible attack.
\begin{proposition}
\label{prop:actuator-containment}
Let $P=P^\top\succ0$ and consider the set $\mathcal E(P)$ in \eqref{eq:ellipsoid}. For every $i\in\mathcal N$, if there exist
$\beta_{i,b},\beta_{i,a}\geq0$ such that for every $\chi \in \mathbb{R}^{5N+1}$ it holds that
\begin{equation}
F_{u,i}(\chi)
\leq \beta_{i,b}F_b(\chi)
+\beta_{i,a}F_a(\chi),
\label{eq:actuator-containment}
\end{equation}
then, for every $q_g\in\mathcal E(P)$ and
$a\in\mathcal A(q_g)$, the injected reactive power satisfies
$|\tilde u_i|\leq\bar Q_i$ for all $i\in\mathcal N$.\hfill $\square$
\end{proposition}
\emph{Proof:}
For every $q_g\in\mathcal E(P)$ and
$a\in\mathcal A(q_g)$, we have $F_b(\chi)\leq0$
and $F_a(\chi)\leq0$. By \eqref{eq:actuator-containment}, and 
since the multipliers $\beta_{i,b}$, $\beta_{i,a}$ are nonnegative, it holds that $F_{u,i}(\chi)
\leq 0$, i.e., 
$|\tilde u_i|\leq\bar Q_i$ for all $i\in\mathcal N$.
\hfill $\blacksquare$
 Proposition \ref{prop:actuator-containment} establishes a capacity guarantee
throughout the candidate ellipsoid. When combined
with Proposition \ref{prop:RPI} for the same matrix $P$, such
guarantee on the inverter capacity holds along every admissible trajectory
initialized in $\mathcal E(P)$.

\subsection{Voltage Safety Certificates}
    The invariant set condition restricts the controller state $q_g$,
whereas the physical safety requirement is imposed on the customer
voltage magnitude $v_i$. For any predefined tolerance $\delta\geq0$, nonnegativity
of $v_i$ implies that the requirement
$|v_i-\bar v|\leq\delta$ is equivalent to
\[
\bigl(\max\{0,\bar v-\delta\}\bigr)^2
\leq v_i^2\leq(\bar v+\delta)^2.
\]
Since $y_{\mathrm{ph},i}=\bar v^2-v_i^2$, we obtain
\begin{equation}
    \underline y
    \leq y_{{\rm ph},i}
    \leq \overline y,
    \label{eq:yinterval}
\end{equation}
where $\underline y:=-(2\bar v\delta+\delta^2)$ and
$\overline y:=\bar v^2-
\bigl(\max\{0,\bar v-\delta\}\bigr)^2$. Note that for $\delta\geq\bar v$, the lower bound on the voltage requirement automatically holds since 
$v_i\geq0$.

 Defining the matrix $C_y=[\,H\;0\;D_y\;I\;0\,]$,
such that $y_{\rm ph}=C_y\chi$, which yields $y_{{\rm ph},i}
    =e_i^\top C_y\chi$, we have that:
\begin{align}
    y_{\rm ph,i}
    &=
    (e^\top\chi)(e_i^\top C_y\chi)=
    \frac{1}{2}
    \chi^\top
    \He(e e_i^\top C_y)
    \chi. \nonumber
\end{align}
Define $B_i^+:=\dfrac12\He(e e_i^\top C_y) -\overline y\,ee^\top$, and
$B_i^-:=-\dfrac12\He(e e_i^\top C_y)
          +\underline y\,ee^\top.$
The upper bound $y_{\rm ph,i}\leq\overline y$
is then equivalent to $\chi^\top B_i^+\chi\leq0$,
while the lower bound $y_{\rm ph,i}\geq\underline y$
is equivalent to $\chi^\top B_i^-\chi\leq0$. Furthermore, let us define the vectors
$\lambda_{i,w}^{\pm}
:=\operatorname{col}(\lambda_{i,w,1}^{\pm},\ldots, \lambda_{i,w,N}^{\pm})$, and $
F_w(\chi):=\operatorname{col}(F_{w,1}(\chi),\ldots, F_{w,N}(\chi)).$
The following result then holds. 
\begin{proposition}
\label{prop:voltage-containment}
Consider the ellipsoid $\mathcal E(P)$, with
$P=P^\top\succ0$, and let $\delta\geq0$.
For every $i\in\mathcal N$, suppose there exist
nonnegative multipliers $\lambda_{i,b}^{\pm}$,
$\lambda_{i,a}^{\pm}$, and
$\lambda_{i,w,j}^{\pm}$, $j\in\mathcal N$, such that
\begin{equation}
B_i^+
\preceq
\lambda_{i,b}^{+}\Phi_b(P)
+\lambda_{i,a}^{+}\Phi_a
+\sum_{j=1}^{N}\lambda_{i,w,j}^{+}\Phi_{w,j},
\label{eq:voltage-upper-LMI}
\end{equation}
\begin{equation}
B_i^-
\preceq
\lambda_{i,b}^{-}\Phi_b(P)
+\lambda_{i,a}^{-}\Phi_a
+\sum_{j=1}^{N}\lambda_{i,w,j}^{-}\Phi_{w,j}.
\label{eq:voltage-lower-LMI}
\end{equation}
Then, for every $q_g\in\mathcal E(P)$ and every
 attack $a \in \mathcal{A}(q_g)$ and disturbance $w \in \mathcal{W}$, the physical
customer voltages $v_i$ satisfy
$\bar v-\delta\leq v_i\leq\bar v+\delta$
for all $i\in\mathcal N$.
\hfill $\square$
\end{proposition}
\emph{Proof:}
For every $q_g\in\mathcal E(P)$ and every admissible
attack and disturbance, it holds that
$F_b(\chi)\leq0$, $F_a(\chi)\leq0$, and
$F_{w,j}(\chi)\leq0$ for all $j\in\mathcal N$.
By the definitions of $B_i^\pm$,
\eqref{eq:voltage-upper-LMI} and
\eqref{eq:voltage-lower-LMI} imply
\[
\left\{
\begin{aligned}
\chi^\top B_i^+\chi
&\leq \lambda_{i,b}^{+}F_b(\chi)
+\lambda_{i,a}^{+}F_a(\chi)
+(\lambda_{i,w}^{+})^\top F_w(\chi)\leq0,\\
\chi^\top B_i^-\chi
&\leq \lambda_{i,b}^{-}F_b(\chi)
+\lambda_{i,a}^{-}F_a(\chi)
+(\lambda_{i,w}^{-})^\top F_w(\chi)\leq0,
\end{aligned}
\right.
\]
where the second inequalities follow from the
nonnegativity of the multipliers.
Thus, it holds that $\underline y\leq y_{\rm ph,i}\leq\overline y$.
By the equivalence in \eqref{eq:yinterval}, it holds that
$\bar v-\delta\leq v_i\leq\bar v+\delta$,
which concludes the proof.
\hfill $\blacksquare$

Proposition \ref{prop:voltage-containment} bounds the physical voltage for
every state $q_g$ in the candidate ellipsoid $\mathcal{E}(P)$ and every
admissible attack and disturbance. Combined with
robust positive invariance of the same ellipsoid,
it guarantees satisfaction of the voltage limits
for all time.
\subsection{Safety Certificate under Attacks}
Propositions~1-3 establish, respectively,  
robust positive invariance, satisfaction of inverter capacity constraints,
and physical voltage safety.
The following theorem combines these results for
a common ellipsoid, yielding a safety certificate
for all admissible trajectories initialized in it.
\begin{theorem}
\label{thm:main-safety}
Under Assumptions 1-4, let $\delta\geq0$
and consider a capacity limit $\bar Q_i >0$ for every $i\in\mathcal N$.
Suppose there exists a common matrix
$P=P^\top\succ0$ satisfying
\eqref{eq:attack-nonempty} such that the following
conditions are satisfied:
\begin{enumerate}
\item  There exist \(V=\operatorname{diag}(\nu_1,\ldots,\nu_N)\succ0\),
    \(\alpha_1\in\mathbb R\), \(\alpha_2\ge0\), and
    \(\alpha_{w,j}\ge0\), \(j\in\mathcal N\), such that \eqref{eq:rpi-condition} holds;
\item For all $i$, there exist
    $\beta_{i,b},\beta_{i,a}\ge0$ such that \eqref{eq:actuator-containment} holds;
\item  
For every $i\in\mathcal N$, there exist nonnegative
multipliers $\lambda_{i,b}^{\pm}$,
$\lambda_{i,a}^{\pm}$, and
$\lambda_{i,w,j}^{\pm}$ for every $j\in\mathcal N$,
such that \eqref{eq:voltage-upper-LMI} and
\eqref{eq:voltage-lower-LMI} hold;
    \end{enumerate}
Then, every
trajectory \(q_g(\cdot)\) under the dynamics \eqref{eq:attacked-state}-\eqref{eq:attacked-output} with \(q_g(0)\in\mathcal E(P)\),
generated by disturbances satisfying $w(t)\in\mathcal W$, and attacks satisfying 
$a(t)\in\mathcal A(q_g(t))$, \text{for all} $t\ge 0$,
satisfies the following:
\begin{enumerate}[label=(\roman*)]
\item  
$q_g(t)\in\mathcal E(P)$, for all \(t\ge0\);
\item $|\tilde u_i(t)|\le\bar Q_i$, for all \(t\ge0\) and all \(i\in\mathcal N\);
\item $\bar v-\delta\le v_i(t)\le\bar v+\delta$ for all \(t\ge0\) .
\end{enumerate}
\hfill \(\square\)
\end{theorem}
\emph{Proof:}
Proposition~\ref{prop:RPI} implies that
$q_g(0)\in\mathcal E(P)$ yields
$q_g(t)\in\mathcal E(P)$ for all $t\ge0$ under every
admissible attack and disturbance. Since the state therefore
remains in $\mathcal E(P)$, Proposition~
\ref{prop:actuator-containment} guarantees
$|\tilde u_i(t)|\le\bar Q_i$, while Proposition~
\ref{prop:voltage-containment} guarantees
$\bar v-\delta\le v_i(t)\le\bar v+\delta$. \hfill $\blacksquare$

The certificate conditions are LMIs for fixed values of the
scalar multipliers multiplying $\Phi_b(P)$. As such, joint optimization
of these multipliers and $P$ introduces bilinear terms such as
$\alpha_1P$, $\beta_{i,b}P$, and
$\lambda_{i,b}^{\pm}P$. The next subsection normalizes the capacity and
voltage inequalities so that only $\alpha_1$
is handled by an outer search.

\subsection{Linear Matrix Inequality Feasibility Problem}
To obtain a formulation that is easier to calculate, we restrict attention to certificates with
$\beta_{i,b}>0$ and $\lambda_{i,b}^{\pm}>0$ for all $i \in \mathcal{N}
$.  Under this restriction, for each fixed $\alpha_1=a_1\in\mathbb{R}$, the conditions of
Theorem~1 can be reformulated as an LMI feasibility problem. Specifically, we introduce the variables $t_i := \dfrac{1}{\beta_{i,b}}$,
$s_i := \dfrac{\beta_{i,a}}{\beta_{i,b}}$,
and $t_i^\pm := \dfrac{1}{\lambda_{i,b}^\pm}$,
$s_{i,a}^\pm := \dfrac{\lambda_{i,a}^\pm}{\lambda_{i,b}^\pm}$,
$s_{i,w,j}^\pm :=
\dfrac{\lambda_{i,w,j}^\pm}{\lambda_{i,b}^\pm}$. For compactness, let us now define
$\boldsymbol{\alpha}_w
:=\operatorname{col}(\alpha_{w,1},\ldots,\alpha_{w,N})$,
$\boldsymbol t:=\operatorname{col}(t_1,\ldots,t_N)$, $\boldsymbol s:=\operatorname{col}(s_1,\ldots,s_N)$,
$\boldsymbol t^\pm:=\operatorname{col}(t_1^\pm,\ldots,t_N^\pm)$,$\boldsymbol s_a^\pm :=\operatorname{col}(s_{1,a}^\pm,\ldots,s_{N,a}^\pm)$,
$\boldsymbol s_{i,w}^\pm :=\operatorname{col}(s_{i,w,1}^\pm,\ldots,s_{i,w,N}^\pm).$
For any $z\in\mathbb R^N$, let
$\mathcal M_w(z):=\sum_{j=1}^{N}z_j\Phi_{w,j}.$
All vector inequalities below are interpreted
componentwise, and indexed constraints hold for
every $i\in\mathcal N$. Since $\beta_{i,b}>0$ and $\lambda_{i,b}^\pm>0$,
dividing the corresponding matrix inequalities by these
multipliers preserves their definiteness and yields the
following equivalent feasibility problem:
\begin{equation}
\left\{
\begin{aligned}
\text{Find}\;&
P,\;V,\;\alpha_2,\;\boldsymbol{\alpha}_w,\;
\boldsymbol t,\;\boldsymbol s, \boldsymbol t^\pm,\;\boldsymbol s_a^\pm,\;
\{\boldsymbol s_{i,w}^\pm\}_{i\in\mathcal N}
\\[1mm]
\text{s.t.}\;&
P=P^\top\succ0, V=\operatorname{diag}(\nu_1,\ldots,\nu_N)\succ0,\\
&\alpha_2\geq0,\quad \boldsymbol{\alpha}_w\geq0,  \boldsymbol  t>0,\quad \boldsymbol t^\pm>0,\\
&\boldsymbol s\geq0,\quad \boldsymbol s_a^\pm\geq0,\quad
  \boldsymbol s_{i,w}^\pm\geq0,\\
&Q_1-Q_2Q_3^{-1}Q_2^\top\preceq P,\\
&\Phi_0(P)+a_1\Phi_b(P)-\alpha_2\Phi_a-\Phi_s(V)-\mathcal M_w(\boldsymbol{\alpha}_w)
  \preceq0,\\
&t_i\Phi_{u,i}\preceq\Phi_b(P)+s_i\Phi_a,\\
&t_i^\pm B_i^\pm\preceq
  \Phi_b(P)+s_{i,a}^\pm\Phi_a+\mathcal M_w(\boldsymbol s_{i,w}^\pm).
\end{aligned}
\right.
\label{eq:lmi-feasibility}
\end{equation}
The original multipliers in Theorem~1 are recovered as $\beta_{i,b}=\dfrac{1}{t_i}$,
$\beta_{i,a}=\dfrac{s_i}{t_i}$,
and $\lambda_{i,b}^\pm=\dfrac{1}{t_i^\pm}$,
$\lambda_{i,a}^\pm=\dfrac{s_{i,a}^\pm}{t_i^\pm}$,
$\lambda_{i,w,j}^\pm=\dfrac{s_{i,w,j}^\pm}{t_i^\pm}$.
Thus, only $a_1$ must be now tuned
externally. The remaining multipliers are found by solving the
LMI feasibility problem~\eqref{eq:lmi-feasibility}.

\section{Numerical Simulations}
We illustrate the safety verification framework on the five-customer low voltage radial distribution network considered in \cite{chong2019}, based on the CIGRE benchmark
in \cite{strunz2009}. The model parameters are listed in \autoref{tab:system_parameters}. We set the nominal voltage
to $\bar v=230\,\mathrm{V}$ and the safety tolerance to
$\delta=11.5\,\mathrm{V}$, corresponding to a $\pm5\%$
voltage margin. Each customer employs the piecewise droop
function in \eqref{eq:droop_piecewise}, which satisfies
Assumption~\ref{assum:sector-bounded}, with saturation
limits $\bar K_i=\bar Q_i$ matching the inverter
reactive power capacities. The inverter time constants are set to $\tau_i=1\,\mathrm{s}$
for all $i\in\mathcal N$. The inverter reactive power limits in Table~I are computed
as $\bar Q_i=\sqrt{\bar s_i^2-\rho_{g,i}^2}$, using
$\bar s=[4200 \ 6500 \ 4700 \ 5300\ 3600]^\top\,\mathrm{VA}$
and
$\rho_g=[3500 \ 5500 \ 4000 \ 4500 \ 3000]^\top\,\mathrm{W}$.  We set $y_{m,i}=y_{n,i}=0$ and
$y_{\max,i}=-y_{\min,i}=14899.4\,\mathrm{V}^2$.
Lemma~2 then applies with $S_1=0$ and
$S_2=\operatorname{diag}(d_1,\ldots,d_N)$,
where $d_i=\bar Q_i/y_{\max,i}$. We consider an operating envelope satisfying
Assumption~\ref{ass:prosumer_bounds}, with the net active
power and reactive power consumption bounds listed in
Table~I, and take $\Delta_0=5\,\mathrm{V}$ in Assumption~3.
Applying Corollary~1 gives
$\bar w=10^3[2.6353\ 2.9346\ 2.9388\ 3.1131\ 3.1186]^\top
\,\mathrm{V}^2$. We consider coordinated attacks $a=\Gamma q_g+m$, with
$\Gamma=\operatorname{col}(\Gamma_y,\Gamma_u)$,
$\Gamma_y=-0.5I_5$ in $\mathrm{V}^2/\mathrm{VAr}$, and
$\Gamma_u=-0.5I_5$. The independent component satisfies
$m^\top Q_m m\leq 1$, where
\[
Q_m=\operatorname{diag}\!\left(
\bar m_{y,1}^{-2},\ldots,\bar m_{y,N}^{-2},
\bar m_{u,1}^{-2},\ldots,\bar m_{u,N}^{-2}
\right),
\]
with $\bar m_{y,i}=1000\,\mathrm{V}^2$ and
$\bar m_{u,i}=0.7\bar Q_i$. Thus, Assumption~4 holds with $Q_1=\Gamma^\top Q_m\Gamma$,
$Q_2=-\Gamma^\top Q_m$,
$Q_3=Q_m$.
The actuator feedback gives $\tilde u=0.5q_g+m_u$.
For the simulation study, we consider the 
disturbance $w_i(t)=\bar w_i\sin(2\pi\cdot0.1t+\varphi_i)$,
where $\varphi_i=i\,\mathrm{rad}$ for each $i\in\mathcal N$, and choose the attack bias to be
$m_{y,i}(t)=\dfrac{\bar m_{y,i}}{\sqrt{5}}\sin(2\pi\cdot0.4t)$
and
$m_{u,i}(t)=\dfrac{\bar m_{u,i}}{\sqrt{5}}\cos(2\pi\cdot0.4t)$ for each customer $i\in\mathcal N$.
Note that these signals satisfy $w(t)\in\mathcal W$ and
$m(t)^\top Q_m m(t)=1$ for all $t\geq0$, where the latter implies that the attack signal satisfies Assumption \ref{ass:attack_set}. 

\begin{table}[b]
\centering
\renewcommand{\arraystretch}{1.2}
\caption{Network Parameters and Bounds}
\begin{tabular}{|c|c|c|c|c|c|}
\hline
$i$ & $1$ & $2$ & $3$ & $4$ & $5$ \\
\hline
$R_{i-1}\,[\Omega]$
& 0.00343 & 0.00172 & 0.00343 & 0.00515 & 0.00172 \\
\hline
$X_{i-1}\,[\Omega]$
& 0.04711 & 0.02356 & 0.04711 & 0.07067 & 0.02356  \\
\hline
$R'_{i-1}\,[\Omega]$
& 0.00147 & 0.00662 & 0.00147 & 0.00147 & 0.00147 \\
\hline
$X'_{i-1}\,[\Omega]$
& 0.02157 & 0.09707 & 0.02157 & 0.02157 & 0.02157 \\
\hline
$\bar\rho_i\,[\mathrm{W}]$
& 1205 & 60 & 1440 & 2205 & 280 \\
\hline
$\bar q_{c,i}\,[\mathrm{VAr}]$
& 300 & 960 & 480 & 600 & 400 \\
\hline
$\bar Q_i\,[\mathrm{VAr}]$
& 2321.6 & 3464.1 & 2467.8 & 2800 & 1990 \\
\hline
\end{tabular}
\vspace{0.2cm}
\label{tab:system_parameters}
\end{table}

\begin{figure}[h]
    \centering
    \includegraphics[width=0.82\linewidth]{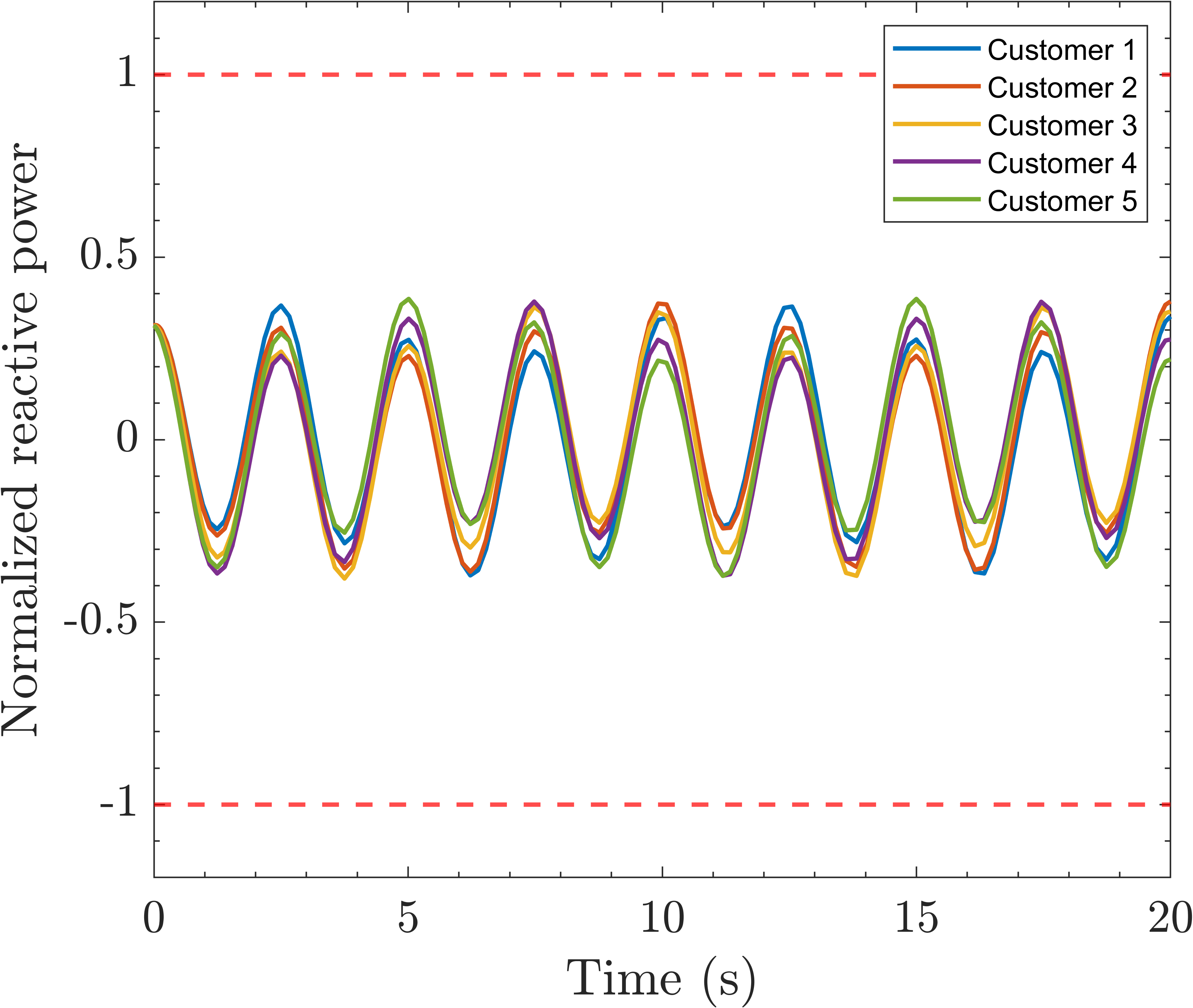}
    \caption{Normalized reactive power injections after actuator
corruption. Dashed lines indicate the inverter capacity limits.}
    \label{fig:normalized_u_tilde}
\end{figure}
\begin{figure}[h]
    \centering
    \includegraphics[width=0.8\linewidth]{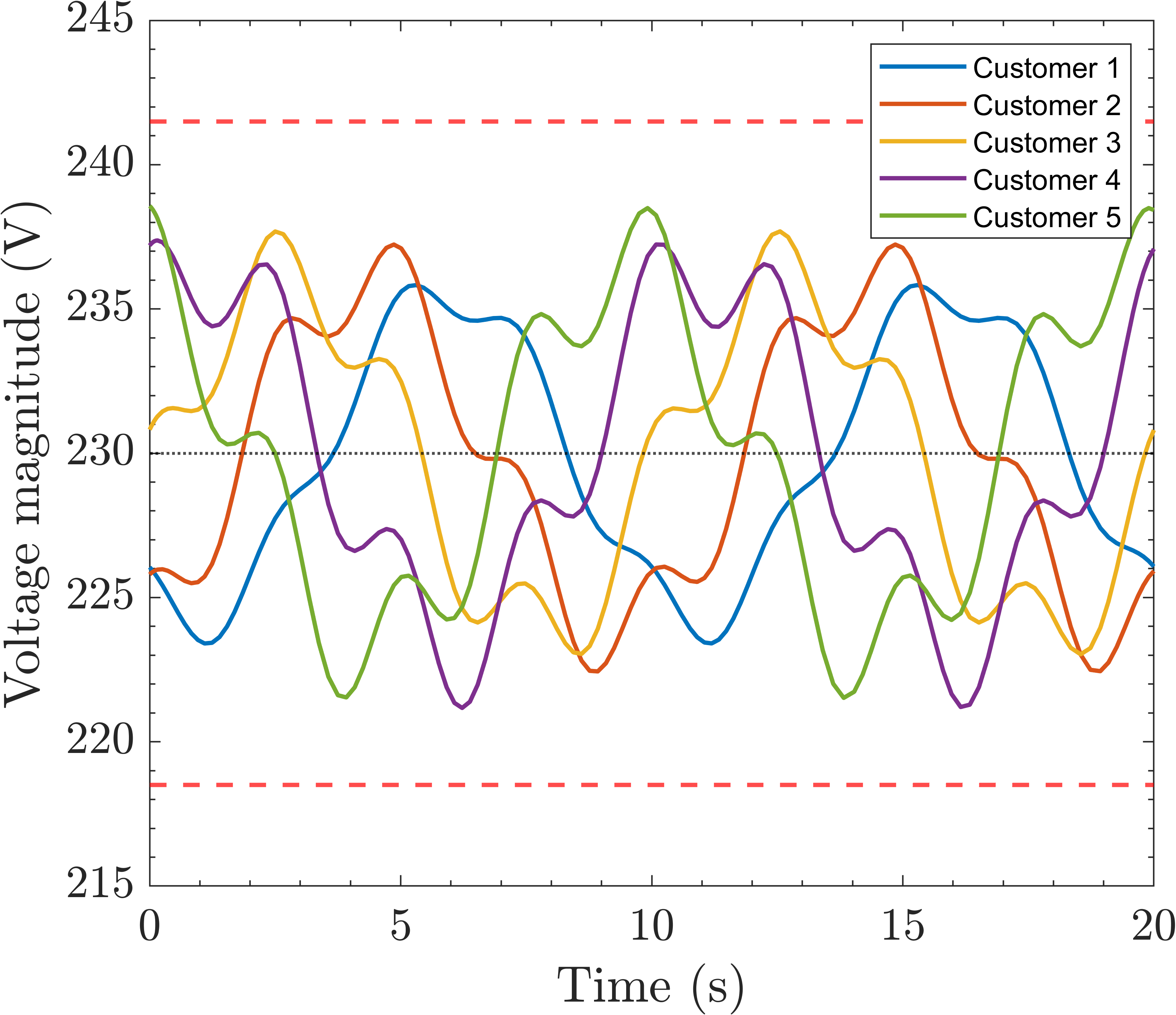}
    \caption{Customer voltage magnitudes under coordinated attacks.
Dashed lines indicate the predetermined $\pm5\%$ voltage limits.}
    \label{fig:voltages_attack}
\end{figure}

A feasible certificate is obtained at $a_1 \approx 1.429$,
selected through a grid search. The returned solution is then confirmed against the original
certificate inequalities.
Figures~3 and~4 illustrate the system response over a time interval $t\in[0,20]\,\mathrm{s}$ with initial conditions $q_g(0)=0$. 
Figure~3 shows the normalized reactive power injections
$\tilde u_i/\bar Q_i=(q_{g,i}+a_{u,i})/\bar Q_i$. Note that
all normalized reactive power injections remain within $[-1,1]$, demonstrating that
the inverter capacities are respected even after actuator
corruption. This illustrates the role of the capacity
containment condition, i.e., bounding the controller state alone would not otherwise  ensure that the attacked injections remain feasible. Figure~4 shows that all customer voltage
magnitudes remain within the predetermined interval
$[218.5,241.5]\,\mathrm{V}$ despite the coordinated sensor
and actuator attacks in accordance to Theorem \ref{thm:main-safety}.

\section{Concluding Remarks}
We developed a safety verification framework for
droop-controlled radial distribution feeders under simultaneous sensor and actuator attacks.
The certificate accounts for attacks and disturbances entering the nonlinear droop control and jointly ensures
reactive power state invariance, customer voltage safety and inverter capacity
compliance. Future work will investigate the design of a secondary
controller that allows
to maintain or improve voltage and inverter capacity limits under
a broader set of admissible attacks and disturbances.

\bibliographystyle{IEEEtran}
\bibliography{references}

\end{document}